\documentclass[12pt,reqno]{article}

\usepackage{amssymb}

\usepackage{fullpage}

\usepackage{amsmath}
\usepackage{amsthm}
\usepackage{amsfonts}

\usepackage{cyrillic}

\begin{document}

\newtheorem{theorem}{Theorem}
\newtheorem{corollary}[theorem]{Corollary}
\newtheorem{lemma}[theorem]{Lemma}
\newtheorem{proposition}[theorem]{Proposition}
\newtheorem{conjecture}[theorem]{Conjecture}
\newtheorem{defin}[theorem]{Definition}
\newenvironment{definition}{\begin{defin}\normalfont\quad}{\end{defin}}
\newtheorem{examp}[theorem]{Example}
\newenvironment{example}{\begin{examp}\normalfont\quad}{\end{examp}}
\newtheorem{rema}[theorem]{Remark}
\newenvironment{remark}{\begin{rema}\normalfont\quad}{\end{rema}}

\title{Inverse Star, Borders, and Palstars}

\author{Narad Rampersad\\
Department of Mathematics\\
University of Li\`ege\\
Grande Traverse, 12 (Bat. B37)\\
4000 Li\`ege\\
Belgium\\
{\tt narad.rampersad@gmail.com} \\
\and
Jeffrey Shallit\\
School of Computer Science\\
University of Waterloo\\
Waterloo, ON  N2L 3G1\\
Canada\\
{\tt shallit@cs.uwaterloo.ca}
\and
Ming-wei Wang\\
Microsoft Corporation\\
Redmond, WA  \\
USA\\
{\tt m2wang@gmail.com} \\
}

\maketitle

\begin{abstract}
A language $L$ is closed if $L = L^*$.
We consider an operation on closed languages, $L^{-*}$,
that is an inverse to Kleene closure.  It is known that if
$L$ is closed and
regular, then $L^{-*}$ is also regular.  We show that
the analogous result fails to hold
for the context-free languages.  Along the way we find a new relationship
between the unbordered words and the prime palstars of Knuth, Morris, and
Pratt.  We use this relationship to enumerate the prime palstars, and we
prove that neither the language of all unbordered words nor the language of
all prime palstars is context-free.
\end{abstract}

\section{Inverse star}

Let $L$ be a language such that $L = L^*$.  Then,
following
\cite{Brzozowski&Grant&Shallit:2009}, we say that $L$ is {\it closed}.
Brzozowski \cite{Brzozowski:1967} studied the
the ``smallest'' language $M$ such that $L = M^*$.  

\begin{definition}
For closed languages $L$, define
$$L^{-*} = \bigcap_{S^* = L} S.$$
\end{definition}

Brzozowski proved 

\begin{theorem}
If $L$ is closed then $(L^{-*})^* = L$.  Furthermore
$L^{-*} = L - L^2$.  If $L$ is regular and closed, then so is
$L^{-*}$.
\end{theorem}

In this note we show that the class of context-free languages is not closed
under the operation $-*$.  First, though, we take a digression
to discuss products of palindromes.

\section{Palstars, prime palstars, and unbordered words}
\label{palstar}

In this section we find a new connection between the prime palstars
(as introduced in Knuth, Morris, and Pratt \cite{Knuth&Morris&Pratt:1977})
and the unbordered words.

We start with some definitions.  By $w^R$ we mean the reverse of the
word $w$.  A {\it palindrome} is a word
$w$ such that $w = w^R$.  In this paper we will only be concerned
with the nonempty palindromes of even length:

$${\tt PAL} = \lbrace x x^R \ : \ x \in \Sigma^+ \rbrace .$$

A {\it palstar} is an element of the language ${\tt PALSTAR} := {\tt PAL}^*$.

A word $x$ is a {\it prime palstar} if it is a palstar and cannot
be written as the product of two palstars.  Evidently a prime palstar
must itself be a palindrome.  The first few prime
palstars over $\lbrace 0,1 \rbrace$
are $00, 0110, 010010, 011110, 01000010, 01011010,
01111110, $ and their complements, obtained by mapping $0$ to $1$
and vice versa.  The language of all prime palstars is denoted
{\tt PRIMEPALSTAR}.

\begin{theorem}[Knuth-Morris-Pratt \cite{Knuth&Morris&Pratt:1977}]
Every palstar has a unique factorization into prime
palstars.
\end{theorem}

The proof of this theorem depends on the following lemma:

\begin{lemma}[Knuth-Morris-Pratt \cite{Knuth&Morris&Pratt:1977}]
No prime palstar is a proper prefix of another prime palstar.
\label{pal}
\end{lemma}

\begin{corollary}
If $w$ is a palindrome of even length, then its factorization
into prime palstars must be of the form
$w = x_1 x_2 \cdots x_n$, where $x_i = x_{n+1-i}$ for 
$1 \leq i \leq n$.
\label{palin}
\end{corollary}

\begin{proof}
Suppose $w = x_1 \cdots x_n$ is the factorization into prime palstars 
$x_i$.  If $n = 1$ we are done.  
Otherwise, since $w$ ends with $x_n$, it must begin with $x_n^R = x_n$.  
Hence either $x_1$ is a prefix of $x_n$, or vice versa.
By Lemma~\ref{pal} we must have $x_1 = x_n$.  Using the same argument
on the shorter palindrome ${x_1}^{-1} w x_1^{-1}$, we derive the
remaining equalities.
\end{proof}

We now turn to borders.  A word is said to be {\it bordered} if
it has some nonempty prefix that is also a suffix.  Otherwise,
it is {\it unbordered}.  Unbordered words are also called
{\it bifix-free} in the literature \cite{Nielsen:1973}.  

Equivalently,
a word $w$ is bordered if it can be written in the form
$xyx$ for some nonempty word $x$.  For example, 
{\tt entanglement} begins and ends with the string
{\tt ent}.  

Given two words of the
same length $x = a_1 a_2 \cdots a_n$ and $y = b_1 b_2 \cdots b_n$,
their {\it perfect shuffle} $x \sha y$ is defined by
$x \sha y = a_1 b_1 \cdots a_n b_n$.

\begin{theorem}
A word $w$ is a prime palstar if and only if there exists an
unbordered word $z$ such that $w = z \sha z^R$.
\label{palstar-thm}
\end{theorem}

\begin{proof}

Suppose $w$ is not a prime palstar.  If $w$ is not an even length
palindrome then it is certainly not of the form $z \sha z^r$.
Suppose then that $w$ is an even length palindrome and hence is of
the form $z \sha z^R$.  We will show that $z$ is bordered.  Since $w$ is
not a prime palstar we can factor $w$ into a product of prime
palstars.  Then by Corollary~\ref{palin} such a factorization must
look like $x \cdots x$ for some palindrome $x$.  Then when we
``unshuffle" $w$ into $z$ and $z^R$, we get that $z$ starts with the
odd-indexed letters of $x$ and ends with the odd-indexed letters of
$x^R$.  But $x = x^R$, so $z$ starts and ends with the same word.

On the other hand, suppose $w = x \sha y$.
By comparing the symbols $x$ to $y$ we see that if $y \not= x^R$,
then $w$ is not a palindrome.    So assume $y = x^R$.  Now if $x$
is bordered, then we can write it as $x = zuz$ for some nonempty string
$z$.  Then $w = (zuz) \sha (zuz)^R = (z \sha z^R) (u \sha u^R)
(z \sha z^R)$.  
This gives a factorization of $w$ as a product of two or three nonempty
palstars (according to whether $u$ is empty or nonempty).
\end{proof}

An example of this theorem in English is
${\tt noon}$, which is a prime palstar, and is the shuffle of the unbordered
word
${\tt no}$ with its reversal.

\section{Enumeration of palstars}

As far as we know, up to now
no one has enumerated the palstars.  However, our
argument above allows us to do so, based on enumeration of the
unbordered words.

Nielsen \cite{Nielsen:1973} has shown that if $a_n$ denotes the number of unbordered words
of length $n$ over an alphabet of size $k$, then
$$
a_n = \begin{cases}
	k, & \text{if $n = 1$;} \\
	k a_{n-1} - a_{n/2}, & \text{if $n$ even}; \\
	k a_{n-1}, & \text{if $n$ odd and $> 1$}.
	\end{cases}
$$
(Also see \cite{Blom:1994}.)  Furthermore, he showed that
$a_n \sim c_k k^n$, where $c_k$ is a constant that tends to 
$1$ as $k \rightarrow \infty$, and $c_2 \doteq .2677868$.  

It follows that if $b_n$ is the number of prime palstars of length $2n$,
then $b_n = a_n$.   In particular, about 27\% of all binary 
palindromes are prime palstars.

\section{Context-free languages and inverse star}

We now apply the results in Section~\ref{palstar} to prove that the class
of context-free languages is not closed under inverse star.

Clearly ${\tt PALSTAR} = {\tt PAL}^*$ is context-free.
We have ${\tt PRIMEPALSTAR} = {\tt PALSTAR}^{-*}$.  So it suffices to show that
${\tt PRIMEPALSTAR}$ is not context-free.    Suppose it were.  First, we
need the following result.

\begin{theorem}
The language $U$ of unbordered words over an alphabet of size
at least $2$ is not context-free.
\label{ming}
\end{theorem}

\begin{proof}
Assume it is.  Without loss of generality the alphabet is
$\Sigma = \lbrace 0, 1, \ldots \rbrace$.
Consider
	$$U' := U \cap 1\  0^+ \  1 \  0^+ \ 1 \ 0^+ \ 1 \ 0^+ ,$$
the intersection of $U$ with a regular language.
Then
$$    U' := \lbrace  1\  0^a \ 1 \ 0^b \ 1 \ 0^c \  1\ 0^d \ : \ 
(a < d) {\rm\ and\ } ((a \not= c) {\rm\ or\ } (b < d))  \rbrace .$$
Since the context-free languages are closed under intersection with
a regular language, it suffices to prove $U'$ is not context-free.

To do this, we use Ogden's lemma \cite{Ogden:1968}. Choose
	$$z = \overbrace{1 0^{n+n!}}^A \ \overbrace{1 0^{n+1 + n!}}^B \ 
	\overbrace{1 0^n}^C \ \overbrace{1 0^{n+1+n!}}^D  \in U', $$
and distinguish the third block of $0$'s, the one corresponding to $C$.
Write $z = uvwxy$.  Then by Ogden's lemma
$vwx$ must contain at most $n$ distinguished
positions and $vx$ at least one. 

If $vx$ contains a $1$, then by pumping we get a string with too many $1$'s.
Thus $vx$ contains $0$'s only, and each of $v$, $x$ is contained in a
single block of zeros.

\medskip

\noindent Case 1:  $v$ contains $0$'s from block $A$,
and $x$ contains $0$'s from block $C$.
Then consider $u v^2 w x^2 y =
1\  0^{a'}\  1 \  0^{b'} \ 1\  0^{c'}\ 1 \ 0^{d'}$.
It has $a' \geq d'$, a contradiction.

\medskip

\noindent Case 2:  $v$ contains $0$'s from block $B$,
and $x$ contains $0$'s from block $C$.
Then consider $u v^i w x^i y =
1\  0^{a'}\  1 \  0^{b'} \ 1\  0^{c'}\ 1 \ 0^{d'}$,
where $i = (n!/|x|)+1$.  Then
this string has $a' = c'$, $b' \geq d'$, a contradiction.

\medskip

\noindent Case 3:  $vx$ contains $0$'s from block $C$.  Then as in the previous
case, choose $i = (n!/|vx|) + 1$.  The resulting string has $a' = c'$
and $b' \geq d'$, a contradiction.

\medskip

\noindent Case 4:  $v$ contains $0$'s from block $C$, and $x$ contains $0$'s from block
$D$.
Consider $u v^i w x^i y =
1\  0^{a'}\  1 \  0^{b'} \ 1\  0^{c'}\ 1 \ 0^{d'}$ with $i = 0$
to get $a' \geq d'$, a contradiction.
\end{proof}

Now, using this result, we can prove our last result:

\begin{theorem}
Over an alphabet of two or more letters,
{\tt PRIMEPALSTAR} is not context-free.
\end{theorem}

\begin{proof}
Consider the morphisms $g$ and $h$ defined
as follows:  $g(a) = 00$, $g(b) = 01$, $g(c) = 10$, $g(d) = 11$,
and $h(a) = h(b) = 0$, $h(c) = h(d) = 1$.  Then
the effect of $h \circ g^{-1}$ is to extract the odd-indexed letters
from an even-length word.

Assume that ${\tt PRIMEPALSTAR}$ is context-free.  Then
$h(g^{-1}({\tt PRIMEPALSTAR}))$ would be context-free.  But
by Theorem~\ref{palstar-thm}
$h(g^{-1}({\tt PRIMEPALSTAR})) = U$, the language of unbordered words,
which we have shown in Theorem~\ref{ming} to be non-context-free.
\end{proof}

\end{document}